\documentclass[letterpaper, 10 pt, conference]{ieeeconf}    % Use this line for a4 paper

\IEEEoverridecommandlockouts                              % This command is only needed if 
\usepackage{graphics} % for pdf, bitmapped graphics files
\usepackage{epsfig} % for postscript graphics files
\usepackage{mathptmx} % assumes new font selection scheme installed
\usepackage{times} % assumes new font selection scheme installed
\usepackage{amsmath} % assumes amsmath package installed
\usepackage{amssymb}  % assumes amsmath package installed
\usepackage{xcolor}  
\usepackage{subfigure}
\usepackage{graphicx}
\usepackage{colortbl}
\usepackage{censor}
\usepackage{todonotes}
\usepackage{lipsum}
\usepackage{cite}
\usepackage{url}
\usepackage{multirow}
\usepackage{xcolor}
\usepackage{optidef}
\usepackage{arydshln}
\usepackage{hyperref}
\usepackage{algorithm}

\usepackage{amsmath,amsfonts,amsthm}
\usepackage[font=small,skip=3pt]{caption}
\newtheorem{theorem}{Theorem}
\newtheorem{corollary}{Corollary}[theorem]
\newtheorem{lemma}{Lemma}
\newtheorem*{lemma*}{Lemma}
\newtheorem{proposition}{Proposition}
\newtheorem{assumption}{Assumption}

\usepackage[noEnd=true]{algpseudocodex} % Makes vertical lines instead of "ends"

\DeclareMathOperator*{\argmin}{argmin}

\title{\LARGE \bf
Optimality and Suboptimality of MPPI Control\\ in Stochastic and Deterministic Settings
}%

\author{Hannes Homburger$^{1}$, Florian Messerer$^{2}$,  Moritz Diehl$^{2}$, and Johannes Reuter$^{1}$% <-this % stops a space
\thanks{This research was supported by DFG via projects 504452366 (SPP 2364) and 525018088, and by BMWK via 03EI4057A and 03EN3054B.}
\thanks{$^{1}$ Institute of System Dynamics, HTWG Konstanz - University of Applied Sciences, 78462 Konstanz, Germany
        {\tt\small \{hhomburg, jreuter\} @htwg-konstanz.de}}
    \thanks{$^{2}$ Department of Microsystems Engineering (IMTEK) and Department of Mathematics,
    	University of Freiburg, 79110 Freiburg, Germany
    	{\tt\small \{first.last\} @imtek.uni-freiburg.de}}%
}

\begin{document}

\maketitle
\thispagestyle{empty}
\pagestyle{empty}

%%%%%%%%%%%%%%%%%%%%%%%%%%%%%%%%%%%%%%%%%%%%%%%%%%%%%%%%%%%%%%%%%%%%%%%%%%%%%%%%
\begin{abstract}
Model predictive path integral (MPPI) control has recently received a lot of attention, especially in the robotics and reinforcement learning communities.
 This letter aims to make the MPPI control framework more accessible to the optimal control community. We present three classes of optimal control problems and their solutions by MPPI. Further, we investigate the suboptimality of MPPI to general deterministic nonlinear discrete-time systems. Here, suboptimality is defined as the deviation between the control provided by MPPI and the optimal solution to the deterministic optimal control problem. Our findings are that in a smooth and unconstrained setting, the growth of suboptimality in the control input trajectory is second-order with the scaling of uncertainty. %This level of uncertainty is quantified by a parameter that scales the exploration standard deviation used in MPPI.
 The results indicate that the suboptimality of the MPPI solution can be modulated by appropriately tuning the hyperparameters. We illustrate our findings using numerical examples.

 %  Videos are available under
%\begin{scriptsize}
	%\texttt{ https://www.htwg-konstanz.de/en/research-and-transfer/ institutes-and-laboratories/isd/control-engineering/videos/}.
%\end{scriptsize}

\end{abstract}

%%%%%%%%%%%%%%%%%%%%%%%%%%%%%%%%%%%%%%%%%%%%%%%%%%%%%%%%%%%%%%%%%%%%%%%%%%%%%%%%
\vspace*{-0.10cm}
\section{Introduction and Overview}

Model predictive control (MPC) is an optimization-based state-feedback control technique that computes the control inputs concerning the state trajectory predicted over a finite horizon as the solution of an optimal control problem (OCP). To solve an OCP, often numerical methods are used that sequentially generate local first- or second-order models and then solve them efficiently \cite{Rawlings2017}. These approaches lead to very fast convergence rates and can be applied even to high-dimensional problems. However, they can converge to local minima and have issues in the presence of strong \mbox{nonlinearities}. In contrast, zero-order solvers based on sampling have fewer restrictions on the problem and can be executed in parallel but often lead to higher computational effort by ignoring the first- or second-order information. Recently, a sample-based method called Model Predictive Path Integral (MPPI) control has been investigated and applied in numerous publications, e.g. see the survey~\cite{Kazim2024}. Initially, this method inspired by the path integral framework was used to solve a special class of stochastic closed-loop OCPs \cite{Kappen2005,Theodorou2010}. Later it was extended based on the information-theoretic framework to general nonlinear systems with input noise -- however, based on an open-loop problem formulation \cite{Williams2018}. %In both approaches, the uncertainty of the dynamics is related to the uncertainty of the used sample distribution.   
The certainty-equivalence principle \cite{Anderson2007} %-- a fundamental theorem from linear control theory -- 
states that the solution for linear systems with independent noise and costs is independent of the uncertainty. However, in the literature, the uncertainty-aware MPPI method is often used for deterministic OCPs, and the suboptimality of the solution is tacitly ignored, see e.g. \cite{Zhang.2024,Homburger.2023,Yan.2024,Vlahov.2024}. %or investigated for special cases \cite{Patil.2024,Yi.2024}. 
In this context, noise is intentionally introduced for exploration, representing a fundamental distinction from classical stochastic optimal control, where noise arises from the environment. Recently, an MPPI method for deterministic OCPs was presented \cite{Halder2025}. 
%\vspace*{-0.05cm}
\subsection{Contribution and Outline}
\noindent
We review the MPPI control framework from a stochastic optimal control perspective and show that the problem classes considered in the MPPI literature are stochastically equivalent to input-affine stochastic systems with noise on the input channel.  Subsequently, we investigate the suboptimality of standard MPPI to deterministic OCPs and characterize the problems. We introduce a scaling factor $\beta\in\mathbb{R}^+ $ that scales the standard deviation of the MPPI sampling distribution. We prove that for smooth and unconstrained problems, the suboptimality of the control trajectory is $\mathcal{O}( \beta^2 )$ and the suboptimality of the value function is $\mathcal{O}(\beta^4)$. Structure of this letter: in Section \ref{sec_ocp}, we define different classes of stochastic OCPs and discuss their solution via MPPI in Section~\ref{sec_MPPI}. The suboptimality is investigated in Section \ref{sec_sub} and discussed in Section~\ref{sec_dis}. In Section \ref{sec_exp}, we present numerical examples, followed by a summarizing Section~\ref{sec_con}.
%\vspace*{-0.05cm}
\subsection{Notation and Preliminaries}
\noindent
The set of positive real numbers is 
$
\mathbb{R}^+ = \{ x \in \mathbb{R} \mid x > 0 \}$ and the set of non-negative real numbers is 
$
\mathbb{R}^+_0 = \{ x \in \mathbb{R} \mid x \geq 0 \}$.
 The concatenation of two column vectors $x\in\mathbb{R}^{n_x}$ and $y\in\mathbb{R}^{n_y}$ is denoted by $(x,y):=\left[x^\top,y^\top \right]^\top\in\mathbb{R}^{n_x+n_y}$. 
The partial derivative with respect to a variable $x$ is represented as $\frac{\partial}{\partial x}$, indicating differentiation with respect to the explicit argument of the function. Arguments of a function may be specified using subscripts or parentheses. For instance, the function $f_\beta(x)$ has two arguments, $\beta$ and $x$.
Gradients of functions are denoted by 
    $\nabla_x f(x) = \frac{\partial}{\partial x} f(x)^\top,$
representing the vector of partial derivatives with respect to $x$ and $\nabla^2_x f(x)$ denotes the Hessian.
The multivariate normal distribution of a random variable $v \in \mathbb{R}^{n} $  is expressed as 
$ v \sim \mathcal{N}(\mu, \Sigma),$
where $\mu \in \mathbb{R}^n$ is the mean, and $\Sigma \in \mathbb{R}^{n \times n}$ is the covariance matrix. %, which has to be positive definite denoted by $\Sigma \succ 0$.
The expectation of a random variable $W\in\mathbb{R}^{n}$ sampled from $\mathbb{Q}_\beta$ defined by its probability density function $q_\beta(W)$, is \vspace*{-0.1cm}
\begin{equation*}
    \mathbb{E}_{W \sim \mathbb{Q}_\beta}[W] = \int_{\mathbb{R}^{n}} W q_\beta(W) \, dW.
\end{equation*}

For a function \(f : \mathbb{R}^n \to \mathbb{R}^m\), we write \(f(x) = \mathcal{O}(g(x))\) if and only if there exists a constant \(C \in \mathbb{R}^+\) and a neighborhood \(\mathcal{A}\) of 0 such that $
\forall x \in \mathcal{A}: \|f(x)\| \leq C g(x).$
\\

\section{Stochastic Optimal Control Problems}\label{sec_ocp}
We consider the stochastic nonlinear discrete-time system
\begin{equation*}
    x^+ = f(x,u,w),
\end{equation*}
where the state is $x\in\mathbb{R}^{n_x}$,  the control is $u\in\mathbb{R}^{n_u}$, the random disturbance is $w\in\mathbb{R}^{n_w}$, and without loss of generality $w\sim\mathcal{N}(0,I)$ with unit matrix $I$. Given an initial state $\overline x_0\in\mathbb{R}^{n_x}$, an input trajectory $U=(u_0,u_1,...,u_{N-1})$, and a noise trajectory $W=(w_0,w_1,...,w_{N-1})$, the corresponding state trajectory is determined by
\begin{align}
    x_0(U,W,\overline x_0)&=\overline x_0,\label{eq_x_init}\\
    x_{k+1}(U,W,\overline x_0)&=f(x_k(U,W,\overline x_0),u_k,w_k),\label{eq_x}
\end{align}
 using forward simulation for $k=0,1,...,N-1$.
The performance of a trajectory can be evaluated by the overall cost
\begin{equation}\label{eq_overall_cost_init}
    J_{\overline x_0}(U,W):= \sum_{k=0}^{N-1} L(  x_k(\cdot ),u_k)  + E( x_N(\cdot )) ,
\end{equation}
where the states' dependencies given in \eqref{eq_x_init} and \eqref{eq_x} are omitted for readability, the stage cost function is $L:\mathbb{R}^{n_x}\times \mathbb{R}^{n_u}\rightarrow\mathbb{R}$, and $E:\mathbb{R}^{n_x}\rightarrow\mathbb{R}$ is the terminal cost function. %Different classes of OCPs are presented in the following.
\subsection{Closed-Loop Stochastic Optimal Control}
\noindent
The ideal problem is to minimize the expected overall cost for the closed-loop system. This means that each control $u_k$ is selected with respect to the state $x_k(\cdot )$. The key point is that in $x_k(\cdot )$ all previous realizations of disturbances $w_i$ and control actions $u_i$ for $i=0,..,k-1$ are included. This problem can be formalized in the recursive form 
\begin{equation}\label{eq_CL_OCP}
    V_\mathrm{cls}(\overline x_0) = \min_{u_0}\mathbb{E}_{w_0\sim \mathcal{N}(0,\Sigma)}
    %\min_{u_1}\mathbb{E}_{w_1\sim \mathcal{N}(0,\Sigma)}
    ...\min_{u_{N-1}}\mathbb{E}_{w_{N-1}\sim \mathcal{N}(0,\Sigma)} \;J_{\overline x_0}(U,W),
\end{equation}
that we call the \textit{closed-loop stochastic optimal control problem} (CLS-OCP). Note that this problem typically corresponds to optimization over policies because it reacts to disturbances in an optimal way. 
\subsection{Open-Loop Stochastic Optimal Control}
\noindent
Because the CLS-OCP is typically hard to solve, it is sometimes reasonable to optimize the expected overall cost for the open-loop system instead. This means that the control inputs $u_k$ are chosen independent of the previous realization of the disturbances, but in an optimal way concerning the prior known distribution of the disturbances. This problem can be formalized as 
\begin{equation}\label{eq_SOCP}
   V_\mathrm{ols}(\overline x_0) = \min_{U} \; \mathbb{E}_{W\sim \mathcal{N}(0,\overline \Sigma)} \left[ J_{\overline x_0}(U,W) \right],
\end{equation}
that we call the \textit{open-loop stochastic optimal control problem} (OLS-OCP), where $\overline \Sigma=\mathrm{diag}(\Sigma,\Sigma,...,\Sigma)\in\mathbb{R}^{Nn_w\times Nn_w}$ is the mutual covariance matrix for a whole trajectory $W$ with uncorrelated $w_k$. Usually, the solution to this problem requires less effort, because the optimization concerns only one input trajectory. However, when applied in an MPC fashion, the controller's capability to react to realized disturbances is not modeled, leading to a suboptimal solution for the closed-loop system.
\subsection{Deterministic Optimal Control}
\noindent
In this letter, we investigate the suboptimality of MPPI for solving the deterministic OCP represented by the nonlinear program
\begin{equation}\label{eq_DOCP}
   V_\mathrm{det}(\overline x_0)=\min_{U}\; J_{\overline x_0}(U,0),
\end{equation}
where no uncertainty is considered. We call this problem the \textit{deterministic optimal control problem} (DET-OCP). 
%In this case, due to the absence of noise, there is no difference between the OLS-OCP and the recursive CLS-OCP setting. 
If the computation power is limited or if there is no model of the model's uncertainty, practitioners often use this \textit{nominal} setting. Note that in a smooth and unconstrained setting the DET-OCP solution and the CLS-OCP solution are similar for systems with small uncertainties \cite{Messerer2024}.  
\section{MPPI Solution Methods}\label{sec_MPPI}
In this section, we review the MPPI methods for the OCPs stated in the previous section and we show that these methods are inherently restricted to the overall cost structure \vspace*{-0.15cm}
\begin{equation}\label{eq_overall_cost}
    J_{\overline{x}_0}(U,W) = S_{\overline{x}_0}(U+W)+\frac{1}{2}\lVert U\rVert_2^2,  
\end{equation}
 where the \textit{path costs} with $c:\mathbb{R}^{n_x}\rightarrow \mathbb{R}_0^+$ and $V:=U+W$ are 
 \vspace*{-0.15cm}
\begin{equation}\label{eq_path_cost}
    S_{\overline x_0}(V) := E(  x_N(V,0,\overline x_0))+\sum_{k=0}^{N-1} c(  x_k(V,0,\overline x_0)).
\end{equation}  
Thus, the stage cost in \eqref{eq_overall_cost_init} is restricted to the sum of a generic term dependent only on the state and a quadratic control cost.
\subsection{MPPI Control for the CLS-OCP}\label{sec_MPPI_CLSOCP}
\noindent 
We consider the class of CLS-OCPs, where the discrete-time system dynamics is restricted to input-affine systems
% generalisierbar
\begin{equation}\label{eq_affin_system}
    x^+ = \tilde f(x)+B(x)u+G(x)w,
\end{equation}
 with $B:\mathbb{R}^{n_x}\rightarrow \mathbb{R}^{n_x \times n_u}$, $G:\mathbb{R}^{n_x}\rightarrow \mathbb{R}^{n_x\times n_w}$, $w\sim\mathcal{N}(0,I)$, and the stage cost \vspace*{-0.5cm} \begin{equation}\label{eq_stage_cost}
    L(x,u) = c(x) +\frac{1}{2} u^\top R(x)u
\end{equation}
 with $R:\mathbb{R}^{n_x}\rightarrow \mathbb{R}^{n_u\times n_u}$, $R\succ 0$. The main theoretical results are derived under the following strong assumption. We show that this assumption can interpreted as effectively requiring that the noise enters additively on the control.
\begin{assumption} \label{assumption}
There exists a $\lambda\in\mathbb{R}^+$ such that
    \begin{equation} \label{eq_assumption}\lambda B(x)R(x)^{-1}B(x)^\top=G(x)G(x)^\top. \end{equation}
\end{assumption}
\begin{proposition}\label{proposition}
    We assume the dynamics \eqref{eq_affin_system}, \mbox{$w$ i.i.d.,} the stage cost \eqref{eq_stage_cost}, and Assumption~\ref{assumption} holds. Then an equivalent OCP exists, which is specified by the dynamics
   \vspace*{-0.1cm} \begin{equation}\label{eq_stoch_eq}
    x^+ = \tilde f(x)+\overline B(x)(\overline u+ \overline w),
\end{equation}
and the corresponding overall cost \eqref{eq_overall_cost}.  
\end{proposition}
 \begin{proof}
     Substituting $ B(x)=\overline B(x)R(x)^{\frac{1}{2}}$ into \eqref{eq_assumption}, we obtain $\lambda \overline B(x)\overline B(x)^\top=G(x)G(x)^\top $ and note that $G(x) \mathcal{N} (0, I) = \mathcal{N} \left( 0, G(x)G(x)^\top \right) = \mathcal{N}\left(0, \lambda\bar B(x)  \bar B(x)^\top\right) = \bar B(x) \mathcal{N} (0, \lambda I)$. With $ u=R(x)^{-\frac{1}{2}}\overline u$, the dynamics \eqref{eq_stoch_eq} with $\overline w\sim\mathcal{N}(0,\lambda I)$ and the overall cost ~\eqref{eq_overall_cost} specify an equivalent problem. 
 \end{proof} 
\noindent For this class of CLS-OCPs in a continuous-time setting, a closed-form solution can be derived from an exponential transformation of the corresponding stochastic Hamilton-Jacobi-Bellman (HJB) equation, resulting in a linear second-order partial differential equation (PDE) \cite{Kappen2005}. To obtain a linear second-order PDE, the continuous-time version of Assumption~\ref{assumption} must hold \cite[Sec. II.A]{Williams2017}. This PDE can be evaluated as an expectation via the Feynman-Kac lemma \cite[Thm. 8.2.1]{Oksendal.2003}. While the theoretical foundations are based on a continuous-time setting, to reach a discrete-time solution, the Euler–Maruyama method can be applied as implicitly proposed in \cite[Sec.~IV]{Williams2018}. 
% We can investigate the stochastic equivalent system \eqref{eq_stoch_eq}, and \eqref{eq_dis_ass} reduces to
%\begin{equation}\label{eq_RSigma}
%    \lambda \Sigma^{-1}=R.
%\end{equation} 
 The obtained dynamics \eqref{eq_stoch_eq} are presented as a special case in \cite[Sec.~IV]{Williams2018}, where the solution for the CLS-OCPs is given in compact form as 
\begin{equation*}
    V_\mathrm{cls}(\overline x_0)=-\lambda \log \mathbb{E}_{W\sim \mathcal{N}( 0,\overline \Sigma)}\left[ \exp\left( -\frac{1}{\lambda} S_{\overline x_0}(W) \right) \right]
\end{equation*}
 of weighted uncontrolled trajectories, and the optimal control for $k=0$ is given by
\begin{equation}\label{eq_u_star_cl}
    u^\star_{0,\mathrm{cls}}(\overline x_0)=\frac{\mathbb{E}_{W\sim \mathcal{N}( 0,\overline \Sigma)} \left[ \exp\left( -\frac{1}{\lambda} S_{\overline x_0}(W)\right) w_0  \right]}{\mathbb{E}_{W\sim \mathcal{N}( 0,\overline \Sigma)} \left[ \exp\left( -\frac{1}{\lambda} S_{\overline x_0}(W) \right) \right]}.
\end{equation}
 Note the similarity between \eqref{eq_u_star_cl} and the \textit{softargmin} function, which is well known in the field of machine learning.  
\subsection{MPPI Control for the OLS-OCP}
\noindent 
Now we consider a more general class of nonlinear discrete-time systems with additive noise on the input, given by
\begin{equation}\label{eq_nonlinear_dynamics}
    x^+ =  f(x,u+w,0),\;\;\;\; w\sim\mathcal{N}(0,\Sigma),
\end{equation}
 quadratic costs on the inputs \eqref{eq_stage_cost}, and assumption $\lambda R^{-1}=\Sigma$.
According to~\cite[Sec.~3.A]{Williams2018}, the optimal control sequence $U^\star_\mathrm{ols}$ of the OLS-OCP \eqref{eq_SOCP} can be approximated by 
\begin{equation}\label{eq_KL}
    U^\star_\mathrm{ols} \approx \tilde U^\star_\mathrm{MPPI}:=  \argmin_{U} \mathbb{E}_{W \sim \mathbb{Q}^\star}\left[ \log \left( \frac{q^\star(W)}{q_U(W)}\right) \right], 
\end{equation}
sometimes called \textit{information-theoretic} optimum.
Here, $q_U(W)$ is the probability density function (PDF) of the normal distribution $\mathcal{N}(U,\overline \Sigma)$  and the \textit{optimal distribution} $\mathbb{Q}^\star$ is defined by its PDF
\begin{equation}
\label{eq_opt_dis}
    q^\star(W):=\frac{1}{\eta} \exp \left( -\frac{1}{\lambda} S_{\overline x_0} (W) \right) p_{\overline{\Sigma}}(W),
\end{equation}
where $\eta\in\mathbb{R}^+$ is a normalization factor and  $p_{\overline{\Sigma}}(W)$ is the PDF of the normal distribution $\mathcal{N}(0,\overline \Sigma)$. Note that the objective in \eqref{eq_KL} is the \textit{Kullback–Leibler divergence} \cite[Ch. 3]{Goodfellow.2016} of the optimal distribution $\mathbb{Q}^\star$ and the controlled distribution $\mathbb{Q}_{U}$. In contrast to the presentation in \cite[Eq. (3) and (16)]{Williams2018}, as postulated in \eqref{eq_KL}, MPPI only provides an approximation to the OLS-OCP solution, although the \textit{optimal distribution} \cite[Sec.~3]{Williams2018} establishes a lower bound on the value function of the OLS-OCP \cite{Theodorou.2015}. This follows from the fact that in general, the \textit{optimal distribution} cannot be reached by $\mathbb{Q}_U$. Consequently, the relationship between $U^\star_\mathrm{ols}$ and $\tilde U^\star_\mathrm{MPPI}$ should be interpreted as a heuristic, and it is not clear how suboptimal the relation is.  
The minimizer of \eqref{eq_KL} is given by
\begin{equation}\label{eq_u_star_ol}
    \tilde u^\star_{k,\mathrm{MPPI}}(\overline x_0)=\frac{\mathbb{E}_{W\sim \mathcal{N}( 0,\overline \Sigma)}\left[ \exp\left( -\frac{1}{\lambda} S_{\overline x_0}(W)  \right) w_k\right]}{\mathbb{E}_{W\sim \mathcal{N}( 0,\overline \Sigma)}\left[ \exp\left( -\frac{1}{\lambda} S_{\overline x_0}(W) \right) \right]}%=\mathbb{E}_{\mathbb{Q}^\star}[v_k]
\end{equation}
for $k=0,1,...,N-1 $ \cite[Sec.~3.C]{Williams2018}. In contrast to the presentation in \cite[Sec. 4]{Williams2018}, the equivalence of \eqref{eq_u_star_cl} and \eqref{eq_u_star_ol} for $k=0$ implies an interpretation of \eqref{eq_KL} as an approximation of the CLS-OCP. Here, suboptimality presumably depends on a metric of how strongly the nonlinear dynamics \eqref{eq_nonlinear_dynamics} deviate from the input-affine case \eqref{eq_affin_system}. 
\subsection{MPPI Control for the D-OCP}
\noindent 
However, there is a special case in which MPPI via the \textit{information-theoretic} framework converges to an exact solution:  consider the general nonlinear and deterministic system
\begin{equation*}
    x^+ =f(x,u,0)
\end{equation*}
and stage costs given in the form of \eqref{eq_stage_cost} containing quadratic costs on the control inputs and arbitrary state-dependent costs.
The optimal deterministic input trajectory is given by 
\begin{equation}\label{eq_det_OCP}
   U^\star_\mathrm{det}= \argmin_{U} J_{\overline x_0}(U,0)
\end{equation}
and the $k^\mathrm{th}$ element of this solution is denoted by $u^\star_{k,\mathrm{det}}(\overline x_0)$.
By introducing the scaling parameter $\beta\in\mathbb{R}^+$ and replacing $\lambda$ and $\overline \Sigma $ with $ \beta^2 \lambda $ resp. $ \beta^2 \overline \Sigma$ in \eqref{eq_u_star_ol}, we obtain
\begin{align}
    \tilde u^\star_{k,\mathrm{MPPI}}\left(\overline x_0,\beta\right)&=\frac{\mathbb{E}_{W\sim \mathcal{N}\left( 0,\beta^2 \overline \Sigma\right)}\left[ \exp\left( -\frac{1}{\beta^2\lambda} S_{\overline x_0}(W) \right) w_k \right]}{\mathbb{E}_{W\sim \mathcal{N}\left( 0,\beta^2\overline \Sigma\right)}\left[ \exp\left( -\frac{1}{\beta^2\lambda} S_{\overline x_0}(W) \right) \right]},\notag \\
    &=\mathbb{E}_{W\sim \mathbb{Q}^\star\left(\beta\right)}\left[  w_k  \right],\label{eq_u_star}
\end{align}
and in the limit of shrinking uncertainty this expression with
\begin{equation}\label{eq_u_det_OCP}
    u^\star_{k,\mathrm{det}}(\overline x_0)=\lim_{\beta\rightarrow 0}  \tilde u^\star_{k,\mathrm{MPPI}}(\overline x_0,\beta)
\end{equation}
 provides an exact solution to the DET-OCP for \mbox{$k=0,...,N-1$}.
 Note that the noise is intentionally introduced for exploration and is not part of the environment.
 The statement \eqref{eq_u_det_OCP} is proven in Thm.~1 in Section~\ref{sec_sub}. 

%Note that assumption \eqref{eq_RSigma} still has to be considered. However, it is less restrictive in the deterministic case, since $\lambda$ and $\Sigma$ are only hyperparameters that vanish in the limit.  
\begin{algorithm} 
    \caption{Deterministic MPPI solving \eqref{eq_u_det_OCP}}  \label{alg:DeterministicSTLMPPISolver}
    {\small
	\begin{algorithmic}[1]
        \setlength{\itemsep}{1pt}  % Increase spacing between items
        \Require Initial system state $\overline x_0$, initial control sequence $\hat U$, number of iterations $I$, number of trajectory samples  $M$, shrinking factor $\nu$, initial covariance $\overline \Sigma_0$, initial temperature $\lambda_0$
        \Ensure  Optimal input trajectory $U_\mathrm{det}^\star$%, optimal state trajectory $X_\mathrm{d}^\star$
        \For {$j \in \{0, 1,  \dots, I-1\}$} \label{eq:Det_MPPI_Shrink_Loop_Start}
            \State $\beta \gets \nu^j$; $\lambda \gets \beta^2\lambda_0$;  $\overline \Sigma \gets \beta^2 \overline \Sigma_0$  \label{eq:Det_MPPI_Shrinking} \Comment{reducing uncertainty} 
            \For {$m\in \{0, 1, \dots, M-1\}$ }  \Comment{in parallel} \label{eq:Det_MPPI_MPPI_Loop_Start}
              %  \State ${x}^m_{0} \leftarrow \overline x_0$; $S^m\gets 0$
                \State $W^m \sim \mathcal{N}\left(0,\overline \Sigma\right) $
                \State $S^m \gets S_{\overline x_0}\left(\hat U+W^m\right) + \lambda {W^m}^\top \overline{\Sigma}^{-1} \hat U $;  \label{eq:cost_with_importance_weighting} \Comment{2$^{\mathrm{nd}}$ term is correction because  $\hat U \neq 0$; see Section~\ref{sec_Det_MPPI_Algo}}
                %\For {$k \in \{0, 1, \dots, K-1\}$}
                 % \State  ${x}^m_{k+1} \leftarrow {f}({x}^m_{k}, \hat{u}_{k} + {w}_{k}^m,0)$   
                %\EndFor
               % \State $S^m \gets S^m + \tilde{E}(\tilde{x}^m_{K})$
            \EndFor \label{Det_MPPI_MPPI_Loop_End}
            \State $\psi \gets \min_{m \in \{0, 1, \dots , M-1\}} S^m$ \Comment{offset only for numerics}\label{eq:Det_MPPI_MPPI_Weighting_Start}
            \State $\eta \gets \sum_{m=0}^{M-1} \exp (-\frac{1}{\lambda}(S^m - \psi))$
            \For {$m \in \{0, 1, \dots, M-1\}$} \Comment{in parallel}
                \State $\omega^m \gets \frac{1}{\eta} \exp (-\frac{1}{\lambda}(S^m - \psi))$\Comment{compute weights}
            \EndFor
            \For {$k\in\{0, 1, \dots, K-1\}$}
                \State $\hat u_k \gets \hat u_k + \sum_{m=1}^{M} \omega^m  w_k^m $ \Comment{solution to \eqref{eq_u_star}}
            \EndFor \label{eq:Det_MPPI_MPPI_Weighting_End}
        \EndFor
    %\State $U_\mathrm{det}^\star \gets \hat U$%, x_\mathrm{d,0}^\star \gets x_0$ 
    %\For {$k \in \{0, 1, \dots, K-1\}$}
    %    \State $x^\star_{\mathrm{d},k+1} \gets f(x^\star_{\mathrm{d},k}, u_{k,\mathrm{d}}^\star)$
    %\EndFor
    \State \Return $U_\mathrm{det}^\star \gets \hat U$ \Comment{solution to \eqref{eq_u_det_OCP}}%$U_\mathrm{det}^\star$
    %,X_\mathrm{d}^\star$
    \end{algorithmic}}
\end{algorithm}

\subsection{MPPI Control Algorithms}\label{sec_Det_MPPI_Algo}
The expectation operator regarding the distribution of the uncontrolled system is part of the MPPI solution to the CLS-OCP \eqref{eq_u_star_cl}, the OLS-OCP \eqref{eq_u_star_ol}, and the DET-OCP \eqref{eq_u_det_OCP}. 
MPPI algorithms typically approximate this expectation operator using Monte Carlo estimation \cite{Kazim2024}. However, the sample efficiency using the uncontrolled distribution is strongly limited. To increase sample efficiency with \textit{importance sampling} the sample distribution is changed to a \textit{proposal distribution}. In the same course, a correction term must be added in the expectation to preserve its value. To change the mean of the normal distribution from $0$ to $\hat U$ and sample from $\mathcal{N}(\hat U,\overline \Sigma)$ instead of $\mathcal{N}(0,\overline \Sigma)$, a correction term presented in \cite[Sec. 3.C]{Williams2018} is added to Line \ref{eq:cost_with_importance_weighting} of Algo.~\ref{alg:DeterministicSTLMPPISolver}, which is adapted from \cite{Halder2025}. Further, for $I=1$, Algo.~\ref{alg:DeterministicSTLMPPISolver} reduces to standard MPPI presented in \cite{Williams2018} computing \eqref{eq_u_star_cl} resp. \eqref{eq_u_star_ol}.
It is important to note that importance sampling with this correction term does not change the value of the expectation \cite[Sec.~3.C]{Williams2018}. The suboptimality caused by the Monte Carlo estimation with a finite number of samples is investigated in \cite{Yoon.2022,Patil.2024,Shapiro2009} and is not part of this letter. %For a comprehensive review of general path integral control algorithms see e.g. the survey~\cite{Kazim2024}.\\

\section{Investigation of Suboptimality}\label{sec_sub}
In this section, we present and prove Thm.~\ref{thm_main}, the principal result of this paper. It addresses the convergence of the deterministic MPPI method and establishes the order of suboptimality. These findings are employed in Corollary~\ref{cor} to give an order of suboptimality on the corresponding value function. Before stating Thm.~~\ref{thm_main}, we briefly introduce Lemma~\ref{lemma_laplace} that will be crucial later. %, which is a crucial step in the proof of Theorem~~\ref{thm_main}.
\begin{lemma}[Erdélyi’s formulation of Laplace's classical method, adapted from {\cite[Thm. 1.1]{Nemes.2013}}]\label{lemma_laplace}
For the integral \vspace*{-0.15cm}
\[
I(\lambda) = \int_{a}^{b} \exp \left( -\lambda f(x)\right) g(x) \, \mathrm{d}x,
\]
where $(a,b)\subseteq \mathbb{R}$ is a real interval, which may be finite or infinite, we assume that
\begin{itemize}
    \item[(i)] \( f(x) > f(a) \) for all \( x \in (a, b) \), for every \( \delta > 0 \) the infimum of \( f(x) - f(a) \) in \( [a + \delta, b) \) is positive;
    %\item[(ii)] \( f'(x) \) and \( g(x) \) are continuous in a neighborhood of \( x = a \), except possibly at \( a \);
    \item[(ii)] \(f\) and \(g\) are scalar real analytic functions; and%the expansions (1.2), (1.3), and (1.4) \textit{hold}; and
    \item[(iii)] \textit{the integral} \( I(\lambda) \) \textit{converges absolutely for all sufficiently large} \( \lambda \).
\end{itemize}

\noindent Then
\vspace*{-0.45cm}
\[
I(\lambda) \equiv \exp \left( -\lambda f(a)\right) \sum_{n=0}^{\infty} \tilde c_n\lambda^{-\frac{n}{2}},
\]
where the coefficients \( \tilde c_n \) are real values and the symbol $\equiv$ denotes that the quotient of the left-hand side by the right-hand
side approaches $1$ as \( \lambda \to +\infty \).
\end{lemma}
%In the following parts, we fix the arguments $\overline x_0$ and use $W=0$ by introducing the abbreviated notation $J(U)= J_{\overline x_0}(U,0)$ for readability.
\begin{theorem}\label{thm_main}
Assume $J_{\overline x_0}:\mathbb{R}^{N n_u}\times \mathbb{R}^{N n_u} \rightarrow \mathbb{R}^+_0$ is continuously differentiable, has a unique global minimum  $J_{\overline x_0}(U^\star_\mathrm{det}\;,0)$, and second-order sufficient condition (SOSC) at $U^\star_\mathrm{det}$. %, without loss of generality with value $J_{\overline x_0}\left( U^\star_\mathrm{d},0 \right) = 0$. 
Then 
\begin{equation}\label{eq_thm1}
\lim_{\beta \rightarrow 0}\tilde U^\star_\mathrm{MPPI}(\beta) = U^\star_\mathrm{det}
%\lim_{\beta \rightarrow 0}\mathbb{E}_{W \sim \mathbb{Q}^\star(\beta)} [W] = U^\star_\mathrm{d}
%\tilde U^\star_\mathrm{MPPI}=\mathbb{E}_{W \sim \mathbb{Q}^\star(\beta)} [W]
\end{equation}
 and for small $\beta > 0$ the bias of the solution shrinks quadratically with the uncertainty,
\begin{equation}
\lVert \tilde U^\star_\mathrm{MPPI}(\beta) -U^\star_\mathrm{det} \rVert = \mathcal{O}\left(\beta^2 \right).\label{eq_thm2}
\end{equation}
\end{theorem}
 \begin{proof}
We start with the \textbf{{scalar case}} $Nn_u=1$ and then extend this to the multivariate case. Replacing $\lambda$ and $\overline \Sigma$ by their scaled versions%$\beta^2 \lambda$ and $\beta^2 \overline \Sigma$
, we can rewrite \eqref{eq_opt_dis} as
%\vspace*{0.1cm}
 \begin{align}\label{eq_opt_dis2}
    q_\beta^\star(W):=&\frac{1}{\eta(\beta)} \exp \left( -\frac{1}{\beta^2\lambda } S_{\overline x_0} (W) \right) p_{\beta^2\overline{\Sigma}}(W),\\
    =&\frac{\exp \left( -\frac{1}{\beta^2 \lambda } \left[S_{\overline x_0} (W) + \frac{\lambda}{2} \sum_{k=0}^{N-1} w_k^\top \Sigma^{-1} w_k \right]  \right)}{\eta(\beta)Z(\beta)} ,\notag  \\
    =& \frac{1}{\tilde \eta (\beta)} \exp \left( -\frac{1}{ \beta^2 \lambda  }  J_{\overline x_0}(W,0) \right)\notag ,
\end{align}
where in the second line, the exponential kernel of the multivariate normal distribution is included in the exponent, and its normalization factor is denoted by $Z(\beta)$. In the third line, the term in square brackets is identified as the overall costs \eqref{eq_SOCP}, and the normalization factors are combined to $\tilde{\eta}(\beta)$.  
Based on this representation of the PDF of the optimal distribution, the MPPI solution \eqref{eq_u_star} is given by 
\begin{align}
    \tilde U^\star_{\mathrm{MPPI}}(\beta)&=\mathbb{E}_{W \sim \mathbb{Q}^\star(\beta)} [W] \notag\\
     &=\int_\mathbb{R} W q_\beta^\star(W)\: \mathrm{d} W \notag \\
    &=\frac{\int_\mathbb{R} W \exp \left( -\frac{1}{\beta^2\lambda  } J_{\overline x_0} (W,0) \right) \mathrm{d} W}{\int_\mathbb{R} \exp \left( -\frac{1}{\beta^2\lambda  } J_{\overline x_0} (W,0) \right) \mathrm{d} W},\label{eq_xxx}
\end{align}
where the second line is the explicit representation of the expected value and in the third line $q^\star_\beta (W)$ is substituted.   
For small $\beta>0$, both integrals are dominated by the neighborhood of $W=U^\star_\mathrm{det}$, motivating to apply Laplace's method.
Substituting $\beta^2=\zeta$ in \eqref{eq_xxx} and applying Laplace's method from Lemma~\ref{lemma_laplace} separately to the numerator and the denominator result in
\begin{equation*}
   \tilde U^\star_\mathrm{MPPI}(\zeta)  =  \frac{U^\star_\mathrm{det}\alpha_1 {\zeta}^\frac{1}{2} + \alpha_2 \zeta^{\frac{3}{2}}+\mathcal{O}\left( \zeta^{\frac{5}{2}} \right)}{\alpha_1 {\zeta}^\frac{1}{2}+\alpha_3 {\zeta}^\frac{3}{2}+\mathcal{O} \left( \zeta^{\frac{5}{2}} \right)}, 
\end{equation*}
where the structure follows from SOSC at $U^\star_\mathrm{det}$ and $\alpha_1,\alpha_2,\alpha_3$ are scalar constants.
The re-substitution of $\zeta=\beta^2$ and a subsequent Taylor expansion at $\beta=0$ yields
\begin{equation*}
    \lVert \tilde U^\star_\mathrm{MPPI}(\beta) -U^\star_\mathrm{det} \rVert = \frac{|\alpha_2-U^\star_\mathrm{det}\alpha_3|}{|\alpha_1|}  \beta^2 +\mathcal{O}\left(\beta^4\right)=\mathcal{O}\left( \beta^2 \right).
\end{equation*}
In the \textbf{multivariate case} of \eqref{eq_xxx} the MPPI solution for all scalar elements of the solution trajectory $i=0,..,\hat N-1$ with $\hat N:=Nn_u$ is given by 
\begin{align}
    \tilde U^\star_{\mathrm{MPPI},i}(\beta)&=\mathbb{E}_{w_i \sim \mathbb{Q}^\star(\beta)} [w_i] \notag\\
     & =\frac{\int_{\mathbb{R}} \cdots \int_{\mathbb{R}}\int_{\mathbb{R}}\; w_i\; \Xi(W,\beta)\; \mathrm{d} w_0\mathrm{d} w_1\cdots \mathrm{d} w_{\hat N-1}}{\int_{\mathbb{R}} \cdots \int_{\mathbb{R}}\int_{\mathbb{R}}\; \Xi(W,\beta) \mathrm{d} w_0\mathrm{d} w_1\cdots \mathrm{d} w_{\hat N-1}}\notag \\
     & =\frac{ \int_{\mathbb{R}}  w_i\;  \hat \Xi(w_i,\beta)  \mathrm{d} w_{i}}{\int_{\mathbb{R}}  1 \; \hat \Xi(w_i,\beta)  \mathrm{d} w_{i}}, \label{eq_marginal}
    \end{align}
    where $\Xi(W,\beta):=\exp \left( -\frac{1}{\beta^2\lambda  } J_{\overline x_0} (W,0) \right)$ and $\hat \Xi(w_i,\beta)$ is the marginal PDF for $i=0,..,\hat N-1$. 
     The marginal PDF \eqref{eq_marginal} is similar to \eqref{eq_xxx} that we already treated via Laplace's Method. The same argumentation results in
     \begin{equation*}
    \lVert \tilde U^\star_{\mathrm{MPPI},i}(\beta) -U^\star_{\mathrm{det},i} \rVert = \gamma_i  \beta^2 +\mathcal{O}\left(\beta^4\right)=\mathcal{O}\left( \beta^2 \right)
\end{equation*}
that holds for all scalar elements of input trajectory with index  $i=0,1,...,\hat N-1$ with $\gamma_i\in\mathbb{R}^+_0$ and
\begin{equation}
    \lVert \tilde U^\star_{\mathrm{MPPI}}(\beta) -U^\star_{\mathrm{det}} \rVert = \mathcal{O}\left( \beta^2 \right)
\end{equation}
follows.
 \end{proof}
%\paragraph*{Expository} We suggest to extend Theorem~1 to Lipschitz continuous overall cost functions arising from nonsmooth cost functions or nonsmooth systems. 

% \paragraph*{Scalar Example}
%For the special case of only one prediction step and one control input $N=n_u=1$ the coefficients determined with Laplace's method are given by
%\begin{align*}
%\alpha_1&=1 ,\\
%\alpha_2&=  \frac{5}{24} \frac{ \left(J^{(3)}_{\overline x_0}(U^\star_\mathrm{d},0)\right)^2}{\left(J^{(2)}_{\overline x_0}(U^\star_\mathrm{d},0)\right)^3}  -\frac{1}{8} \frac{J^{(4)}_{\overline x_0}(U^\star_\mathrm{d},0)}{\left(J^{(2)}_{\overline x_0}(U^\star_\mathrm{d},0)\right)^2} 
%,\;\; \text{and}\\
%\alpha_3&= \frac{ \left| J^{(3)}_{\overline x_0}(U^\star_\mathrm{d},0)\right| }{2\left(J ^{(2)}_{\overline x_0}(U^\star_\mathrm{d},0)\right)^2} .
%\end{align*}
%This simple example reveals some interesting properties:
%\begin{itemize}
%    \item For a quadratic $J$ resulting from an LQR setting, $\alpha_2=\alpha_3=0$ holds and the uncertainty does not influence the result. This observation stands in line with the certainty-equivalence principle \cite{Anderson2007}.
%    \item For a continuously differentiable $J$, the error is dominated by the relation of its second- and third-order derivatives. 
%\end{itemize}
%Based on Theorem~1, the following corollary regarding the influence of small uncertainty on the optimal value function can be proven. 
\begin{corollary}\label{cor}
The corresponding optimal value function is reached in the limit and shrinks with fourth-order,
\begin{equation}\label{eq_order_value}
  J_{\overline x_0}(\tilde U^\star_\mathrm{MPPI}(\beta),0)-V_\mathrm{det}(\overline x_0)=  \mathcal{O} \left(\beta^4\right).
\end{equation}
\end{corollary} 
\begin{proof}
    Due to the continuously differentiable setting, the continuous differentiability of the optimal value function is preserved \cite{Messerer2024}.
    Therefore, we can use the Taylor series
    \begin{equation}\label{eq_col}
        J_{\overline x_0}(U,0)-V_\mathrm{det}(\overline x_0)=  \frac{1}{2} \Delta U^\top \nabla^2_U J_{\overline x_0} (U^\star_\mathrm{det},0) \Delta U + \mathcal{O} \left( \lVert U\rVert ^3 \right),
    \end{equation}
    where $\Delta U:= U-U^\star_\mathrm{det}$ denotes the difference to the deterministic solution trajectory.
    Since, with \eqref{eq_thm2}, the error of the computed control trajectory shrinks with $ \mathcal{O}\left( \beta^2 \right) $, and with  \eqref{eq_col} there exists a local quadratic model of the overall cost function, the uncertainty manifests itself with fourth-order in the value function. 
\end{proof}
\section{Alternatives and Computational Effort}\label{sec_dis}
 In comparison to other\textbf{ Monte Carlo optimization} approaches to solve the DET-OCP like reward-weighted regression \cite{Peters2007}, cross-entropy methods \cite{Botev.2013}, covariance matrix adaptation evolution strategy (CMA-ES) \cite{Hansen2006}, \texttt{MPPI-Generic} \cite{Vlahov.2024}, or \texttt{CoVo-MPC} \cite{Yi.2024}, the deterministic MPPI method considered in this paper is characterized by the predefined shrinking rate of the sampling covariance favorable for real-time applications \cite{Halder2025}. However, although it is not part of this study, it should be noted that Monte Carlo approaches often require excessive computational effort.\\% compared to Newton-type optimization. \\
 Methods of \textbf{Newton-type optimization} can be tailored to solve DET-OCPs efficiently \cite[Sec.~8]{Rawlings2017}. To apply Newton-type methods on the OLS-OCP or the CLS-OCP, an approximation of the problem is typically required. 
This takes the form of either sampling the uncertainty vector, cf. e.g., the scenario approach \cite{Calafiore2006} and the sample average approximation \cite{Shapiro2009}, or based on linearization, cf. \cite{Nagy2003a}. %\cite{Diehl2006c}. %Nagy2003a,
For the CLS-OCP, the latter approach requires e.g. a policy parametrization such as linear state  \cite{Nagy2004} %\cite{ Mayne2011} %Nagy2004,
or disturbance feedback \cite{ Goulart2006}.\\ %Ben-Tal2004,
%With structure exploiting algorithms, computation times comparable to those of the DET-OCP can be reached, for the OLS-OCP \cite{Frey2024,Feng2020} %Zanelli2021,,  Frey2024
%as well as state and disturbance feedback \cite{Leeman2024}. %Messerer2021} resp. \cite{ 
%Scenario approaches suffer from exponential complexity with respect to the horizon length, which can be counteracted by only branching the tree a few times \cite{Lucia2014}. 
 Since the problem \eqref{eq_det_OCP} solved by deterministic MPPI follows a single-shooting formulation \cite[Sec. 8.5.1]{Rawlings2017}, which is commonly addressed using Newton-type optimization methods, we compare both approaches in terms of \textbf{computational effort}. The computational cost of evaluating a function \( F \) is denoted as \( \mathrm{cost}(F) \). Due to the parallel nature of sampling, particularly when executed on a GPU \cite{Vlahov.2024}, the computational cost of Algo.~\ref{alg:DeterministicSTLMPPISolver} is given by  
$  \mathrm{cost}(\text{Algo.}~\ref{alg:DeterministicSTLMPPISolver}) \approx I \cdot \mathrm{cost}(S_{\overline{x}_0}),$
as the evaluation of \( S_{\overline{x}_0} \) is computationally expensive. In contrast, the cost of computing the gradient of \( S_{\overline{x}_0} \) using the reverse mode of algorithmic differentiation (AD) is  
$    \mathrm{cost}(\nabla_U S_{\overline{x}_0}) \approx 3 \cdot \mathrm{cost}(S_{\overline{x}_0}) $
\cite[Sec. 8.4.5]{Rawlings2017}. Even Newton-type methods require multiple iterations, leading to an effort of   
   $ \approx3K \cdot \mathrm{cost}( S_{\overline{x}_0})$ to compute only the gradients,
where \( K \) denotes the number of iterations.
From Thm.~\eqref{thm_main} in combination with the exponential shrinking of $\beta$ (Algo.~\ref{alg:DeterministicSTLMPPISolver};~Line~\ref{eq:Det_MPPI_Shrinking}) a \textit{q-linear} convergence follows, while advanced Hessian approximations can lead to faster convergence rates of Newton-type methods restricted to the neighborhood of the optimum  \cite[Thm. 8.7]{Rawlings2017}.

%\paragraph*{Expository}
%We propose to generalize the method investigated in this letter to systems for Lipschitz continuous value functions arising from the application of Lipschitz continuous dynamical systems or cost functions. In this case, the non-smooth value function is smoothed by the uncertainty~\cite{Messerer.2024b}.  

%\begin{theorem}
%Assuming the existence of a unique global minimum at $0$ with value $J(0) = 0$, and $J$ Lipschitz continuous, the following hold:

%\begin{equation*}
%\lim_{\beta \rightarrow 0}\mathbb{E}_{V \sim \mathbb{Q}^\star(\beta)} [V] = 0,
%\end{equation*}

%and for a sufficiently small number $\epsilon > 0$, 

%\begin{equation*}
%\mathbb{E}_{V \sim \mathbb{Q}^\star(\epsilon)} [V] = \mathcal{O}(?),
%\end{equation*}

%where $\mathcal{O}(\epsilon)$ denotes the asymptotic behavior for small $\epsilon$.
%\end{theorem}
 \begin{figure}
 \vspace{0.3cm}
    \centering
    \includegraphics[width=0.85\linewidth]{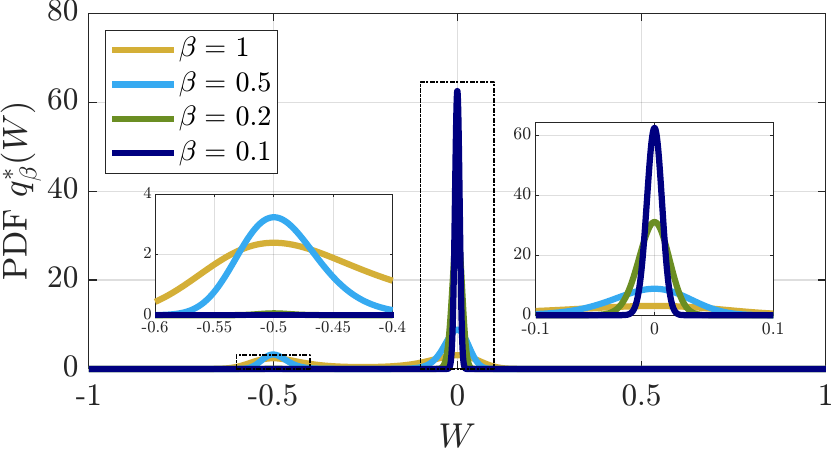}
    \caption{PDF of optimal distribution $\mathbb{Q}^\star_\beta(W)$ for different $\beta$. The two additional interior plots show zoomed sections of the same PDF.}
    \label{fig:PDF}
    \vspace{-0.5cm}
\end{figure}
\begin{figure}[b]
    \centering
    \includegraphics[width=0.95\linewidth]{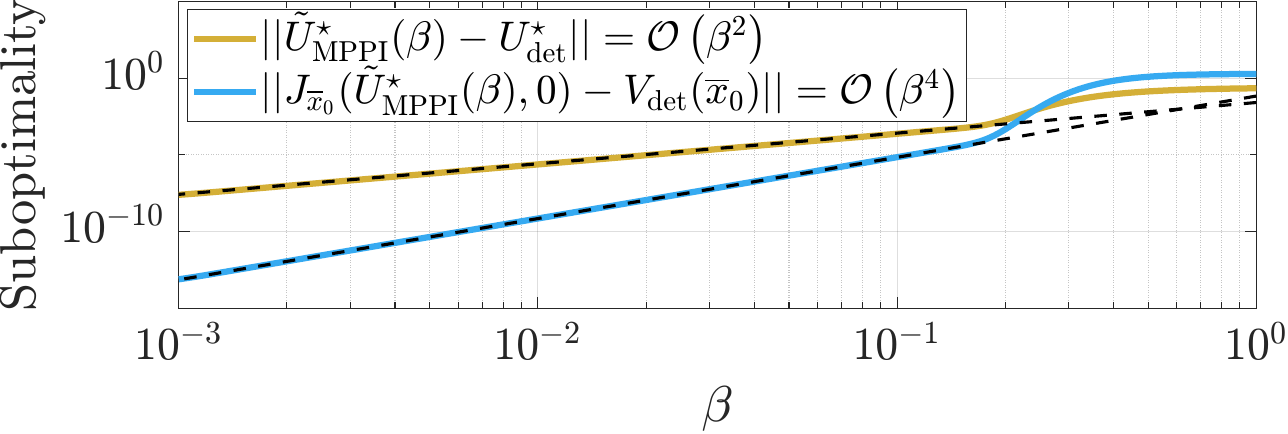}
    \caption{Suboptimality of the MPPI solution to the DET-OCP of the controls (gold) and the value function (blue). }
    \label{fig:SCALAR}
\end{figure}

\section{Numerical Experiments}\label{sec_exp}
  In this section, we provide two numerical experiments, while more complex DET-OCPs and additional technical details are presented in our previous work~\cite{Halder2025}, although without a theoretical investigation of the method's properties. \\
A) Consider the scalar polynomial value function 
%\begin{equation*}
   $ J(U)=\frac{1}{24}c_4 U^4+\frac{1}{6}c_3 U^3+\frac{1}{2}c_2 U^2$
%\end{equation*}
with $U\in\mathbb{R}$ and $c_2,c_3,c_4\in\mathbb{R}$ selected to determine a global minimum at $J(0)=0$, a local minimum at $J(-\frac{1}{2})=\frac{1}{4}$, and the local maximum $J(-\frac{1}{4})=2$.
The PDF of the optimal distribution of this function is plotted in Fig.~\ref{fig:PDF} for different values of $\beta>0$. The suboptimality of the MPPI solution is plotted as a function of $\beta$ in Fig.~\ref{fig:SCALAR}.  
For small $\beta$, the influence of the local minimum at $W=-0.5$ vanishes in Fig.~\ref{fig:PDF} and the corresponding progressions of the error of the MPPI solution show the theoretically predicted behavior.  \\
B) Consider the input-affine dynamics $f_\mathrm{af}=x-\frac{1}{2}\sin(3x)+u+w$ and the nonlinear dynamics $f_\mathrm{nl}=x+\arctan (u+w)$ with $x,u,w\in\mathbb{R}$, where $w\sim\mathcal{N}\left( 0,\beta^2\Sigma \right)$, the stage cost function
$L(x,u)=\frac{1}{2}Ru^2$, the asymmetric terminal cost function $E(x)=(x-1)^6+x$, and a horizon with $N=2$ steps. 
For both settings, the overall cost function has the structure
%\begin{align}
%	\minimize_{x_0,..,x_N,u_0,...,u_{N-1}} \sum_{k=0}^{N-1} & L(x_k,u_k)  + E\left( x_{N}\right) \label{eq_NLP_cost}\\ 
%	\mathrm{subject \; to}    \;\; \;\;\; \;\;\; x_0&=\hat x_0,  \label{eq_NLP_initial}\\
%	x_{k+1}&=F(x_k,u_k),\;\;  k=0,...,N-1,  \label{eq_NLP_dyn}
%\end{align}
\begin{equation}\label{eq_ex2}
    J_{\overline x_0}(U,W)=  L(  x_0(\cdot ),u_0) +  L(  x_1(\cdot ),u_1) + E( x_2(\cdot ))
\end{equation} 
and can be substituted in the CLS-OCP formulation \eqref{eq_CL_OCP}, the OLS-OCP formulation \eqref{eq_SOCP}, and the nominal DET-OCP formulation \eqref{eq_DOCP}. 
 Note that the CLS-OCP solution is visualized as a set because $u_1$ is a policy in this setting.
 The solutions are plotted with the $I=10$ iterates of deterministic MPPI in Fig.~\ref{fig:OCP}. We choose $\overline x_0=-1$, $\hat u_0=\hat u_1=0$, $\nu=\frac{\sqrt 2}{2}$, and a sufficient number of samples. The solutions differ significantly due to the asymmetry of the problem. 
 As considered in Section~\ref{sec_MPPI_CLSOCP}, only in the case of input-affine dynamics, standard MPPI provides a solution to the CLS-OCP. While, in both cases, standard MPPI provides weak approximations to the OLS-OCPs, the iterates of deterministic MPPI converge \textit{q-linearly} to the D-OCP solution.
 
 %The \textit{q-linear} convergence \cite[Sec. 8.3.3]{Rawlings2017} is a design choice caused by the exponential progression of the shrinking (Algorithm~1, Line~\eqref{eq:Det_MPPI_Shrinking}). 
\begin{figure}
\vspace{0.3cm}
    \centering
    \includegraphics[width=0.92\linewidth]{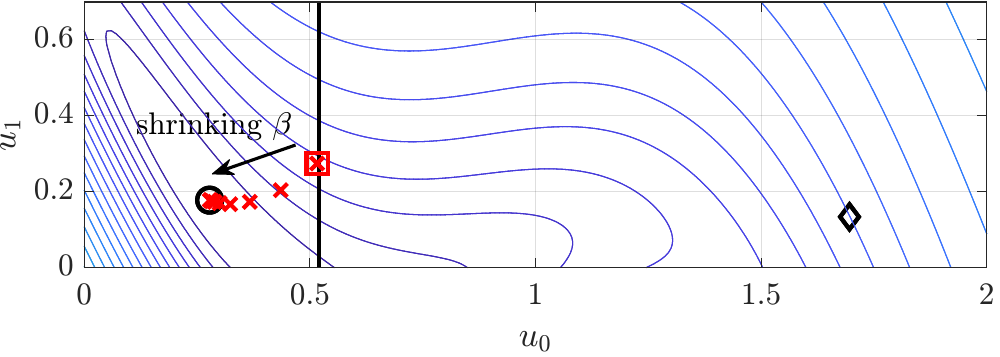}
    \includegraphics[width=0.93\linewidth]{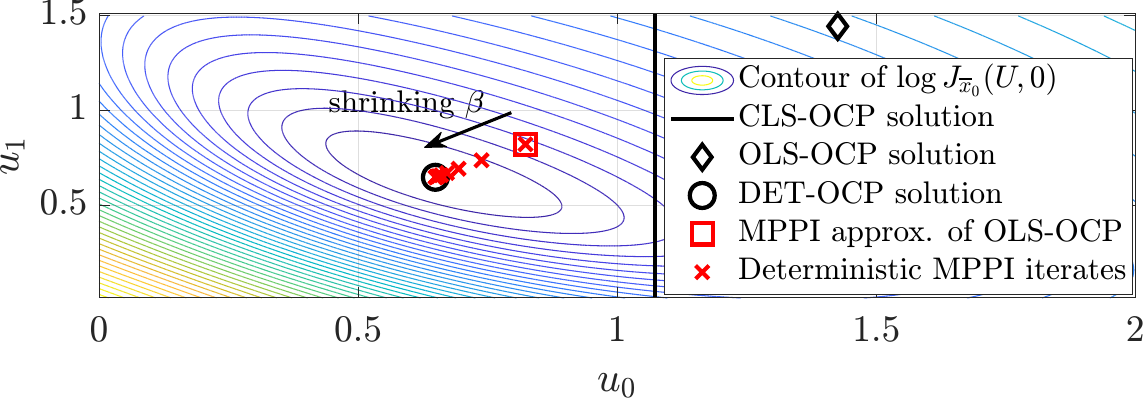}
    \caption{Numerical solutions to the three different problems DET-OCP, OLS-OCP, and CLS-OCP (black), all with overall cost function \eqref{eq_ex2} based on input-affine dynamics $f_\mathrm{af}$ (top) and nonlinear dynamics $f_\mathrm{nl}$ (bottom) alongside the iterates of \textit{deterministic MPPI}.}
    \label{fig:OCP}
    \vspace*{-0.4cm}
\end{figure}

\section{Summary}\label{sec_con}
This letter provides a brief review of different stochastic OCP formulations and their solution methods from the MPPI control framework and simplifies their notation. Further, we show that in a smooth and unconstrained setting the suboptimality of the MPPI solution to deterministic OCPs shrinks smoothly as the level of remaining exploration uncertainty decreases. This proves the convergence of the previously presented deterministic MPPI algorithm and shows that the suboptimality of the MPPI solution can be modulated by appropriately tuning hyperparameters.

\section*{Acknowledgment}
The authors would like to thank Lothar Kiltz for asking about the convergence of deterministic MPPI and Katrin Baumgärtner for valuable discussions on Laplace's method.

\addtolength{\textheight}{-12cm}   % This command serves to balance the column lengths

%%%%%%%%%%%%%%%%%%%%%%%%%%%%%%%%%%%%%%%%%%%%%%%%%%%%%%%%%%%%%%%%%%%%%%%%%%%%%%%%

%%%%%%%%%%%%%%%%%%%%%%%%%%%%%%%%%%%%%%%%%%%%%%%%%%%%%%%%%%%%%%%%%%%%%%%%%%%%%%%%

%%%%%%%%%%%%%%%%%%%%%%%%%%%%%%%%%%%%%%%%%%%%%%%%%%%%%%%%%%%%%%%%%%%%%%%%%%%%%%%%
%\section*{APPENDIX}


%\section*{Acknowledgment}
%Special acknowle to G. Kaddik for determining the parameters of the modeled system dynamics as part of his bachelor thesis, which are listed in Table \ref{tab_Furuta_Parameter}. 
%The preferred spelling of the word acknowledgment in America is without an e after the g. Avoid the stilted expression, One of us (R. B. G.) thanks . . .  Instead, try R. B. G. thanks. Put sponsor acknowledgments in the unnumbered footnote on the first page.

%%%%%%%%%%%%%%%%%%%%%%%%%%%%%%%%%%%%%%%%%%%%%%%%%%%%%%%%%%%%%%%%%%%%%%%%%%%%%%%%

%References are important to the reader; therefore, each citation must be complete and correct. If at all possible, references should be commonly available publications.
%
%
%
%\multirow{2}{*}{}& \multirow{2}{*}{-}&\multirow{2}{*}{-}&\multirow{2}{*}{- } & \multirow{2}{*}{40} & \multirow{2}{*}{3 ms}&\multirow{2}{*}{1 } &Simulation &4.2 \%&25.5&$o$& 5.7 \%&351.7&$\checkmark$\\
%    &&&&&&& Real-World& 71.0\% &14.3&$o$ &  \%  & &\checkmark\\
%
%
%
%\section*{References}
\bibliographystyle{IEEEtran}%ieeetr} % We choose the "plain" reference style

\begin{thebibliography}{}
\providecommand{\url}[1]{#1}
\csname url@samestyle\endcsname
\providecommand{\newblock}{\relax}
\providecommand{\bibinfo}[2]{#2}
\providecommand{\BIBentrySTDinterwordspacing}{\spaceskip=0pt\relax}
\providecommand{\BIBentryALTinterwordstretchfactor}{4}
\providecommand{\BIBentryALTinterwordspacing}{\spaceskip=\fontdimen2\font plus
\BIBentryALTinterwordstretchfactor\fontdimen3\font minus
  \fontdimen4\font\relax}
\providecommand{\BIBforeignlanguage}[2]{{%
\expandafter\ifx\csname l@#1\endcsname\relax
\typeout{** WARNING: IEEEtran.bst: No hyphenation pattern has been}%
\typeout{** loaded for the language `#1'. Using the pattern for}%
\typeout{** the default language instead.}%
\else
\language=\csname l@#1\endcsname
\fi
#2}}
\providecommand{\BIBdecl}{\relax}
\BIBdecl

\end{thebibliography}


\begin{thebibliography}{10}
\providecommand{\url}[1]{#1}
\csname url@samestyle\endcsname
\providecommand{\newblock}{\relax}
\providecommand{\bibinfo}[2]{#2}
\providecommand{\BIBentrySTDinterwordspacing}{\spaceskip=0pt\relax}
\providecommand{\BIBentryALTinterwordstretchfactor}{4}
\providecommand{\BIBentryALTinterwordspacing}{\spaceskip=\fontdimen2\font plus
\BIBentryALTinterwordstretchfactor\fontdimen3\font minus \fontdimen4\font\relax}
\providecommand{\BIBforeignlanguage}[2]{{%
\expandafter\ifx\csname l@#1\endcsname\relax
\typeout{** WARNING: IEEEtran.bst: No hyphenation pattern has been}%
\typeout{** loaded for the language `#1'. Using the pattern for}%
\typeout{** the default language instead.}%
\else
\language=\csname l@#1\endcsname
\fi
#2}}
\providecommand{\BIBdecl}{\relax}
\BIBdecl

\bibitem{Rawlings2017}
J.~B. Rawlings, D.~Q. Mayne, and M.~Diehl, \emph{Model predictive control: {T}heory, computation, and design}, 2nd~ed.\hskip 1em plus 0.5em minus 0.4em\relax Nob Hill Publishing, 2017.

\bibitem{Kazim2024}
M.~Kazim, J.~Hong, M.-G. Kim, and K.-K.~K. Kim, ``Recent advances in path integral control for trajectory optimization: {A}n overview in theoretical and algorithmic perspectives,'' \emph{Annual Reviews in Control}, vol.~57, 2024.

\bibitem{Kappen2005}
H.~J. Kappen, ``Path integrals and symmetry breaking for optimal control theory,'' \emph{J. of Statistical Mechanics: {T}heory and Experiment}, vol. 2005, no.~11, 2005.

\bibitem{Theodorou2010}
E.~A. Theodorou, {Buchli J.}, and S.~Schaal, ``A generalized path integral control approach to reinforcement learning,'' \emph{J. of Mach. learning research}, vol. 2010, no.~11, 2010.

\bibitem{Williams2018}
G.~Williams, P.~Drews, B.~Goldfain, J.~M. Rehg, and E.~A. Theodorou, ``Information-theoretic model predictive control: {T}heory and applications to autonomous driving,'' \emph{IEEE Trans. on Rob.}, vol.~34, no.~6, pp. 1603--1622, 2018.

\bibitem{Anderson2007}
B.~D.~O. Anderson and J.~B. Moore, \emph{Optimal control: Linear quadratic methods}, dover ed., augmented republ. of the ed. 1990~ed., ser. Dover books on engineering.\hskip 1em plus 0.5em minus 0.4em\relax Mineola, NY: {Dover Publications}, 2007.

\bibitem{Zhang.2024}
Y.~Zhang, C.~Pezzato, E.~Trevisan, C.~Salmi, C.~H. Corbato, and J.~Alonso-Mora, ``Multi-modal {MPPI} and active inference for reactive task and motion planning,'' \emph{IEEE Rob. and Autom. Letters}, vol.~9, no.~9, pp. 7461--7468, 2024.

\bibitem{Homburger.2023}
H.~Homburger, S.~Wirtensohn, and J.~Reuter, ``{MPPI} control of a self-balancing vehicle employing subordinated control loops,'' in \emph{2023 European Control Conf. (ECC)}.\hskip 1em plus 0.5em minus 0.4em\relax IEEE, 2023, pp. 1--6.

\bibitem{Yan.2024}
L.~L. Yan and S.~Devasia, ``Output-sampled model predictive path integral control ({o-MPPI}) for increased efficiency,'' in \emph{2024 IEEE Int. Conf. on Rob. and Autom. (ICRA)}.\hskip 1em plus 0.5em minus 0.4em\relax IEEE, 2024, pp. 14\,279--14\,285.

\bibitem{Vlahov.2024}
B.~Vlahov, J.~Gibson, M.~Gandhi, and E.~A. Theodorou, ``{MPPI-G}eneric: A {CUDA} library for stochastic optimization.''\hskip 1em plus 0.5em minus 0.4em\relax arXiv, 2024, DOI:10.48550/arXiv.2409.07563.

\bibitem{Halder2025}
P.~Halder, H.~Homburger, L.~Kiltz, J.~Reuter, and M.~Althoff, ``Trajectory planning with signal temporal logic costs using deterministic path integral optimization,'' in \emph{accepted to the IEEE Int. Conf. on Rob. and Autom.}, 2025.

\bibitem{Messerer2024}
F.~Messerer, K.~Baumg{\"a}rtner, S.~Lucia, and M.~Diehl, ``Fourth-order suboptimality of nominal model predictive control in the presence of uncertainty,'' \emph{IEEE Control Systems Letters}, vol.~8, pp. 508--513, 2024.

\bibitem{Williams2017}
G.~Williams, A.~Aldrich, and E.~A. Theodorou, ``Model predictive path integral control: {F}rom theory to parallel computation,'' \emph{J. of Guidance, Control, and Dynamics}, vol.~40, no.~2, pp. 344--357, 2017.

\bibitem{Oksendal.2003}
B.~{\O}ksendal, \emph{Stochastic Differential Equations}.\hskip 1em plus 0.5em minus 0.4em\relax Berlin, Heidelberg: {Springer Berlin Heidelberg}, 2003.

\bibitem{Goodfellow.2016}
I.~Goodfellow, Y.~Bengio, and A.~Courville, \emph{Deep learning}, ser. Adaptive computation and Mach. learning.\hskip 1em plus 0.5em minus 0.4em\relax Cambridge, Massachusetts and London, England: {The MIT Press}, 2016.

\bibitem{Theodorou.2015}
E.~Theodorou, ``Nonlinear stochastic control and information theoretic dualities: Connections, interdependencies and thermodynamic interpretations,'' \emph{Entropy}, vol.~17, no.~5, pp. 3352--3375, 2015.

\bibitem{Yoon.2022}
H.-J. Yoon, C.~Tao, H.~Kim, N.~Hovakimyan, and P.~Voulgaris, ``Sampling complexity of path integral methods for trajectory optimization,'' in \emph{2022 American Control Conf. (ACC)}.\hskip 1em plus 0.5em minus 0.4em\relax IEEE, 2022, pp. 3482--3487.

\bibitem{Patil.2024}
A.~Patil, G.~A. Hanasusanto, and T.~Tanaka, ``Discrete-time stochastic {LQR} via path integral control and its sample complexity analysis,'' \emph{IEEE Control Systems Letters}, vol.~8, pp. 1595--1600, 2024.

\bibitem{Shapiro2009}
A.~Shapiro, D.~Dentcheva, and A.~Ruszczynski, \emph{Lectures on Stochastic Programming: Modelling and Theory}.\hskip 1em plus 0.5em minus 0.4em\relax SIAM, 2009.

\bibitem{Nemes.2013}
G.~Nemes, ``An explicit formula for the coefficients in {L}aplace's method,'' \emph{Constructive Approximation}, vol.~38, no.~3, pp. 471--487, 2013.

\bibitem{Peters2007}
J.~Peters and S.~Schaal, ``Reinforcement learning by reward-weighted regression for operational space control,'' in \emph{Proc. of the Int. Conf. on Mach. learning}, 2007, pp. 745--750.

\bibitem{Botev.2013}
Z.~I. Botev, D.~P. Kroese, R.~Y. Rubinstein, and P.~L'Ecuyer, ``The cross-entropy method for optimization,'' in \emph{Handbook of Statistics - Mach. Learning: Theory and Applications}, ser. Handbook of Statistics.\hskip 1em plus 0.5em minus 0.4em\relax Elsevier, 2013, vol.~31, pp. 35--59.

\bibitem{Hansen2006}
N.~Hansen, ``The {CMA} evolution strategy: A comparing review,'' in \emph{Towards a New Evolutionary Computation}.\hskip 1em plus 0.5em minus 0.4em\relax Berlin, Heidelberg: Springer, 2006, vol. 192, pp. 75--102.

\bibitem{Yi.2024}
Z.~Yi, C.~Pan, G.~He, G.~Qu, and G.~Shi, ``{CoVO-MPC}: Theoretical analysis of sampling-based mpc and optimal covariance design.''

\bibitem{Calafiore2006}
G.~C. Calafiore and M.~C. Campi, ``The scenario approach to robust control design,'' \emph{IEEE Trans. Automat. Control}, 2006.

\bibitem{Nagy2003a}
Z.~Nagy and R.~Braatz, ``{R}obust nonlinear model predictive control of batch processes,'' \emph{{AIChE} J.}, vol.~49, no.~7, pp. 1776--1786, 2003.

\bibitem{Nagy2004}
------, ``{O}pen-loop and closed-loop robust optimal control of batch processes using distributional and worst-case analysis,'' \emph{J. of Process Control}, vol.~14, pp. 411--422, 2004.

\bibitem{Goulart2006}
P.~J. Goulart, E.~C. Kerrigan, and J.~M. Maciejowski, ``Optimization over state feedback policies for robust control with constraints,'' \emph{Automatica}, vol.~42, pp. 523--533, 2006.

\end{thebibliography}

\end{document}